\documentclass[a4paper,USenglish]{lipics-v2016}
\usepackage{color}
\usepackage{ifthen}
\usepackage{microtype}
\usepackage{thmtools}

\newcommand{\techRep}{true} %% switch here between true and false
\newcommand{\iftechrep}{\ifthenelse{\equal{\techRep}{true}}}

\declaretheorem[sibling=theorem,style=theorem]{proposition}

\newenvironment{qlemma}[1]{%
{\par\mbox{}\newline\noindent\bf Lemma #1.}%
\begin{itshape}%
}{%
\end{itshape}%
}
\newenvironment{qcorollary}[1]{%
{\mbox{}\newline\noindent\bf Corollary #1.}
\begin{itshape}%
}{%
\end{itshape}%
}

\definecolor{orange3}{rgb}{1.0,0.2538,0.1681}
\definecolor{blau}{rgb}{0,.39608,0.74118}
\definecolor{rot}{rgb}{0.79216,.12941,0.24706}
\definecolor{raducolor}{rgb}{.5,.3,.1}

\renewcommand\vec[1]{{\bf #1}}

\newcommand{\g}{\vec{g}}

\title{Proving the Herman-Protocol Conjecture}
\author[1]{Maria Bruna}
\author[2]{Radu Grigore}
\author[3]{Stefan Kiefer}
\author[3]{Jo\"el Ouaknine}
\author[3]{James Worrell}
\affil[1]{Mathematical Institute, University of Oxford}
\affil[2]{School of Computing, University of Kent}
\affil[3]{Department of Computer Science, University of Oxford}

\authorrunning{M. Bruna and R. Grigore and S. Kiefer and J. Ouaknine and J. Worrell} %mandatory. First: Use abbreviated first/middle names. Second (only in severe cases): Use first author plus 'et. al.'

\Copyright{Maria Bruna and Radu Grigore and Stefan Kiefer and Jo\"el Ouaknine and James Worrell}%mandatory, please use full first names. LIPIcs license is "CC-BY";  http://creativecommons.org/licenses/by/3.0/

% mandatory: Please choose ACM 1998 classifications from http://www.acm.org/about/class/ccs98-html . E.g., cite as "F.1.1 Models of Computation". 
% mandatory: Please provide 1-5 keywords
\EventEditors{Ioannis Chatzigiannakis, Michael Mitzenmacher, Yuval
Rabani, and Davide Sangiorgi}
\EventNoEds{4}
\EventLongTitle{43rd International Colloquium on Automata, Languages,
and Programming (ICALP 2016)}
\EventShortTitle{ICALP 2016}
\EventAcronym{ICALP}
\EventYear{2016}
\EventDate{July 11--15, 2016}
\EventLocation{Rome, Italy}
\EventLogo{}
\SeriesVolume{55}
\ArticleNo{XXX}
\begin{document}
%---------------------
\maketitle

\newcommand{\E}{\mathbb{E}}
\newcommand{\Ex}[1]{\mathbb{E}\left[#1\right]}
\newcommand{\N}{\mathbb{N}}
\renewcommand{\P}[1]{{\cal P}\left(#1\right)}
\newcommand{\R}{\mathbb{R}}
\newcommand{\tim}{T}

\begin{abstract}
  Herman's self-stabilization algorithm, introduced 25 years ago, is a
  well-studied synchronous randomized protocol for enabling a ring of
  $N$ processes collectively holding any odd number of tokens to reach
  a stable state in which a single token remains. Determining the
  worst-case expected time to stabilization is the central outstanding
  open problem about this protocol. It is known that there is a
  constant $h$ such that any initial configuration has expected
  stabilization time at most $h N^2$. Ten years ago, McIver and Morgan
  established a lower bound of $4/27 \approx 0.148$ for $h$, achieved
  with three equally-spaced tokens, and conjectured this to be the
  optimal value of $h$. A series of papers over the last decade
  gradually reduced the upper bound on $h$, with the present record
  (achieved in 2014) standing at approximately $0.156$. In this
  paper, we prove McIver and Morgan's conjecture and establish that
  $h = 4/27$ is indeed optimal.
\end{abstract}

%\noindent \textbf{Keywords:} Self-stabilization; Distributed Protocol;
%Probabilistic Protocol; Random Walk; Expected Time

%\newpage
%\pagenumbering{arabic}
%\setcounter{page}{1}

%\input{intro} %from ICALP'11
\section{Introduction} \label{sec-prelim}

The notion of \emph{self-stabilization} was introduced in a seminal
paper of Dijkstra~\cite{D74}, and rose to prominence a decade later,
following (among others) an invited talk of Lamport during which he
pointed out that ``self-stabilization [is] a very important concept in
fault tolerance''~\cite{Lam84}. Both self-stabilization and fault
tolerance have since become central themes in distributed computing
(see, e.g.,~\cite{Dolev}), as recently witnessed by the award of the
2015 Edsger W. Dijkstra Prize in Distributed Computing to Michael
Ben-Or and Michael Rabin for ``starting the field of fault-tolerant
randomized distributed algorithms'' in the early 1980s.

In this paper, we examine an early self-stabilization algorithm known
as Herman's Protocol~\cite{Herman90}, whose exact mathematical
analysis has proven remarkably challenging over the two-and-a-half
decades since its inception. This algorithm considers a ring of $N$
processes (or nodes), where each process either holds or doesn't hold
a token. Starting from any initial configuration of $K$ tokens, where
$K$ is required to be odd, Herman's algorithm proceeds as follows: at
each time step, every process that holds a token either keeps it or
passes it to its clockwise neighbor with probability~$1/2$. All
updates happen synchronously, and if a process finds itself with two
tokens (having simultaneously kept one and received one from its
counterclockwise neighbor) then both tokens are annihilated. It is
straightforward to see that, starting from an odd number of tokens and
following this procedure, almost surely only one token eventually
remains, at which point the ring is said to have \emph{stabilized}.

%Herman's algorithm, as proposed in the original paper~\cite{Herman90},
%can be implemented as follows.
Herman's original paper~\cite{Herman90} presents the algorithm in
a form amenable to implementation. Each process possesses a bit, which the
process can read and write. Each process can also read the bit of its
counterclockwise neighbor.  In this representation, having the same
bit as one's counterclockwise neighbor is interpreted as having a
token. At each time step, each process compares its bit with the bit
of its counterclockwise neighbor; if the bits differ, the process
keeps its bit, whereas if the bits are the same, the process flips its
bit with probability~$1/2$ and keeps it with
probability~$1/2$.
%  It is straightforward to verify that this
%procedure implements Herman's algorithm:
It is straightforward to verify that the bit-flipping version is
an implementation of the token-passing version:
in particular, a process
flipping its bit corresponds to passing its token to its clockwise
neighbor.  If the number of processes is odd, by construction this
bit representation forces the number of tokens to be odd as well,
which justifies the assumption that $K$, the number of tokens, is
always odd. In this paper we make no assumption about the parity of the
number of \emph{processes}, as we abstract from the bit
implementation, and simply assume that the number of tokens is odd
throughout.

%
%\footnote{Notice that flipping all bits in a given configuration keeps
%  all tokens in place.  In fact, in the original
%  formulation~\cite{Herman90}, in each iteration each bit is
%  effectively flipped once more, so that flipping the bit means
%  keeping the token, and keeping the bit means passing the token.  The
%  two formulations are equivalent in the synchronous version, but our
%  formulation allows for an asynchronous version.}

Herman's original paper~\cite{Herman90} showed that the expected time
(number of synchronous steps) to stabilization is $O(N^2 \log N)$.
The same paper also mentions an improved upper bound of $O(N^2)$ due
to Dolev, Israeli, and Moran, without giving a proof or a further
reference.  In 2004, Fribourg et al.~\cite{Fribourg05} established an
upper bound of $2N^2$, and the following year Nakata~\cite{Nakata05}
gave a tighter upper bound of $0.936N^2$ and exhibited an initial
configuration with expected stablization time $\Omega(N^2)$.  At the
same time and independently, McIver and Morgan showed
in~\cite{McIverMorgan} that the initial configuration consisting of
three equally-spaced tokens has an expected stabilization time of
exactly $\frac{4}{27}N^2$, and conjectured that this value is an upper
bound on the expected time to stabilization starting from \emph{any}
initial configuration with \emph{any} (odd) number of tokens. The
conjecture is intriguing since increasing the initial number of tokens
might be thought to lengthen the expected time to stabilization, due
to the larger number of collisions required to achieve stabilization.

Nevertheless, McIver and Morgan's Herman-Protocol Conjecture is
supported by considerable amount of experimental
evidence~\cite{PrismHerman}, and in the intervening years a series of
papers have gradually reduced the upper bound on the constant $h$ such
that stabilization from any initial configuration takes expected time
at most $hN^2$: upper bounds of approximately $0.64$, $0.521$,
$0.167$, and $0.156$ are given respectively in
\cite{KMO11,FZ13,FengZhang14,Haslegrave14}, the last one provided last
year by Haslegrave, and coming relatively close to McIver and Morgan's
lower bound of $4/27 \approx 0.148$.

In this paper, we prove McIver and Morgan's conjecture and establish
that $h = 4/27$ is indeed optimal. Writing $T_z$ for the stabilization
time starting from an initial configuration~$z$, we seek to prove that
$\E T_z \leq \frac{4}{27}N^2$. To this end, one of the key ideas is to
work with a Lyapunov function $V(z)$ in lieu of the (more complicated)
function $\E T_z$.
The domain of the function~$V$ is \emph{continuous}:
a domain element describes a configuration in terms of the distances between adjacent tokens.
Combinatorial arguments exploiting the highly
symmetrical structure of $V(z)$ enable us to establish that, for an
arbitrary configuration~$z$, we have $\E T_z \leq V(z)$, with equality
holding for all three-token configurations.
%A significant difficulty
%to overcome in such an approach---if it is to succeed in establishing
%the optimal bound for $h$---is that our over-approximating proxy
%$V(z)$ must necessarily remain \emph{continuous} at simplex boundaries
%(cf.~Lemma~\ref{lem-f-properties}).
Finally, in what constitutes the most technically challenging part of this paper, 
we combine induction on the number of tokens with analytical techniques to show that $V$ is bounded by $\frac{4}{27}N^2$.
Taken together,
%these various properties of our Lyapunov function $V$ then entail 
we obtain $\E T_z \le \frac{4}{27}N^2$, entailing the Herman-Protocol Conjecture.
%has no local maxima greater than
%$\frac{4}{27}N^2$ on `interior' configurations
%(Lemma~\ref{lem-f-max-induction}). 

The case of there being an \emph{even} number~$K$ of tokens is equally
natural from a mathematical point of view, although it does not
correspond to a concrete bit-flipping protocol.  It was established
in~\cite{FengZhang14} that the worst-case configuration in this
variant is the equidistant \emph{two}-token configuration, with an
expected stabilization time of~$\frac12 N^2$; the analysis underlying
that result is considerably simpler than what is required in case the
number of tokens is odd, as in the present paper.

%.  The proof for the even
%case is based on a suitably chosen quadratic Lyapunov function.  This
%function and the analysis of its maxima are considerably simpler than
%what is needed in the present paper.

Herman's protocol is also related to the notion of \emph{coalescing random
  walks}~\cite{Arratia-Coal-and-Anni81,Cox-Coal-89,AF02}.
There, one considers multiple independent random walks on~$\mathbb{Z}^d$ (or on the vertices of a connected graph).
When two walks meet, they coalesce into a new random walk. %, and the question considered is what is the expected time until all walks coalesce into a single random walk.
A protocol for self-stabilizing mutual exclusion based on such random walks was proposed in~\cite{IsraeliJalfon}.
The expected coalescence time was studied in~\cite{Coppersmith93,Oliveira-12,Cooper-Coalescing}.

It is interesting to note that Herman's ring is closely related to
widely-studied models of random walks and Brownian motion in
statistical physics. Observe that by a simple modification of the
formalism, one may equivalently view Herman's model as a ring
in which tokens randomly move in discrete step in \emph{any}
direction, with pairwise collisions leading to annihilation; this
precisely corresponds to Fisher's \emph{vicious drunks}
model~\cite{Fis84} (with periodic boundary conditions). Similar models
have been studied in chemical physics~\cite{Gen68,Balding88,War91} and
statistical mechanics~\cite{GP06,RS11,SMC13}, among others.

The rest of the paper is organized as follows.  In
Section~\ref{sec-previous} we review previous results in the
literature that are relevant to our proof.  In
Section~\ref{sec-overview} we outline the structure of our proof,
identifying two key lemmas, Lemma~\ref{lem-incr-f5} and
Lemma~\ref{lem-f-max-induction}.  Those are proved in
\iftechrep{Appendix~\ref{sec:incr-f5-proof}}{\cite{BGKOW15-techrep}} and
Section~\ref{sec-f-max-induction}, respectively.

Another solution of the conjecture, using different techniques, is independently shown in~\cite{Endre}.

\section{Relevant Previous Results} \label{sec-previous} 

For the rest of the paper we fix the number~$N$ of processes.  We
assume that the number $K$ of tokens is odd, and both $N$ and $K$ are
at least $3$.

Processes are numbered from $1$ to~$N$, clockwise, according to their
position in the ring.  A configuration with $K$ tokens is formalized
as a function $z : \{1, \ldots, K\} \to \{1, \ldots, N\}$ with
$z(1) < \cdots < z(K)$, where the $i$th token
($i \in \{1, \ldots, K\}$) is held by the processor with the
number~$z(i)$.  We write $Z_K$ for the set of configurations with $K$
tokens, and $Z$ for the set of all possible configurations, that is,
$Z = Z_1 \cup Z_3 \cup Z_5 \cup \ldots$ %\cup Z_{N-2} \cup Z_N$.
%\maria{I think I would write until $M$ only, given that $M$ is the initial number of tokens as defined above.}
%\stefan{I'm trying to de-emphasize $M$ as the number of tokens in the *initial* configuration.
%It could be the number of tokens of an arbitrary configuration.}
 %Note that $Z_1$ is the set of final configurations.

For a fixed initial configuration~$z = z_0$ we write $(z_t)_{t \ge 0}$
for the stochastic process of configurations emanating from~$z$.  The
\emph{stabilization time} $\tim_z$ is the smallest $t \ge 0$ such that
$z_t\in Z_1$, i.e., the time until only one token is left.
%We write $\tim$ for $\tim_z$ if $z$ is unimportant or understood.
In this paper we focus on the expectation~$\E \tim_z$.
It is shown in~\cite{McIverMorgan} that if $N$ is odd and a multiple of~$3$,
then there is a configuration~$z \in Z_3$ 
(with the 3 tokens maximally separated in an equilateral triangle)
such that $\E \tim_z = \frac{4}{27} N^2$.
%By~\cite{Haslegrave14} we have $\E \tim_z \le \frac{\pi^2-8}{12} N^2$
% for all configurations~$z$.

In this paper we show:
\begin{theorem} \label{thm-main}
We have $\E \tim_z \le \frac{4}{27} N^2$ for all $z \in Z$.
\end{theorem}
%\maria{Shouldn't it be Herman's conjecture?}
%\stefan{I guess not, because Herman didn't express it. I called it Herman conjecture because it's a conjecture about Herman's protocol.}
Equivalently, the Herman conjecture states that for all odd $K\ge3$ and all $z \in Z_K$ we have $\E \tim_z \le \frac{4}{27} N^2$.
Only the case $K=3$ was previously known \cite{McIverMorgan}.
%Intuitively, the bound $\frac{4}{27} N^2$ is optimal since the worst 3-token configuration can be ``approximated'' by a $K$-token configuration by forming a block of $K-2$ tokens close together and placing the remaining two tokens in an approximately equilateral triangle with the block. 

The following proposition has been used in a similar form in various papers on Herman's protocol, for instance in~\cite[Lemma~5]{McIverMorgan}.
It bounds the stabilization time by a Lyapunov function~$V$.
\begin{proposition}[Bound by a Lyapunov function] \label{prop-lyapunov}
Given $z \in Z$, denote by $z' \in Z$ the random successor configuration of~$z$.
Let $V : Z \to \R$ be a function with %$V(z) = 0$ for all $z \in Z_1$ and
\begin{align}
  \E(V(z') \mid z) \;&\le\; V(z) - 1
  &&\text{for all $z \in Z \setminus Z_1$, and}
  \label{eq-prop-lyapunov-step}
\\
  0 \;&\le\; V(z)
  &&\text{for all $z \in Z_1$.}
  \label{eq-prop-lyapunov-base}
\end{align}
Then $\E \tim_z \le V(z)$ for all $z \in Z$.
In particular, $V(z) \ge 0$ for all $z \in Z$.
\end{proposition}
Although this result is not new, we give a short proof based on a martingale argument.
The proof is inspired by~\cite{Haslegrave14}, and may provide some intuition.
\begin{proof}
Let $z \in Z$.
Consider the stochastic process $(z_t)_{t \ge 0}$ of configurations emanating from~$z=z_0$. Define $W_t := V(z_t) + t$.
By~\eqref{eq-prop-lyapunov-step} the process $(W_t)_{t \ge 0}$ is a supermartingale.
The stabilization time $\tim_{z} = \tim_{z_0}$ is a stopping time with finite expectation,
and the differences $|W_{t+1} - W_t|$ are bounded as the Markov chain reachable from~$z$ has finitely many states.
Hence, the optional stopping theorem applies, yielding
$
 \E W_{\tim_z} \le \E W_0 = V(z)
$.
By definition of~$W_t$ we have $\E W_{\tim_z} = \E V(z_{\tim_z}) + \E \tim_z$.
Since $z_{\tim_z}\in Z_1$,
  we have $\E \tim_z \le \E W_{\tim_z}$ by~\eqref{eq-prop-lyapunov-base}.
By combining the previous two inequalities, we obtain $\E \tim_z \le V(z)$.
%Further we have $z_{\tim_z} \in Z_1$ and so $V(z_{\tim_z}) = 0$.
%It follows $\E \tim_z = \E X(\tim_z) \le V(z)$. 
\end{proof}

Following~\cite{FengZhang14,Haslegrave14} we associate with a configuration $z \in Z_K$ 
the \emph{gap vector} $\vec{g}(z) = (g_0, \ldots, g_{K-1}) \in \N^K$ by setting 
$g_0 := N+z(1)-z(K)$,
and $g_i := z(i+1) - z(i)$ for $i \in \{1, \ldots, K-1\}$.
Then $\vec{g}(z)/N$ lives in the so-called standard $(K-1)$-simplex $D^{(K)}$, defined by 
\[
 D^{(K)} := \left \{\vec{x} = (x_0, \ldots, x_{K-1}) \in [0,1]^K \ \mid \ x_0 + \cdots + x_{K-1} = 1\right \}.
\]
Towards a suitable Lyapunov function~$V$ we define the cubic polynomial $f_3^{(K)} : D^{(K)} \to [0, \infty)$ by 
\[
  f_3^{(K)}(\vec{x}) :=
    \hskip-2em
    \sum_{\substack{0\le i_0<i_1<i_2<K\\\text{$i_2-i_1$, $i_1-i_0$ odd}}}
    \hskip-2em
       x_{i_0} x_{i_1} x_{i_2}.
%f_3^{(K)}(\vec{x}) := \sum_{i_0=0}^\infty x_{i_0} \cdot \left( 
%          \sum_{i_1=0}^\infty x_{i_0 + 2 i_1 + 1} \cdot \left(
%          \sum_{i_2=0}^\infty x_{i_0 + 2 i_1 + 2 i_2 + 2} \right)\right) \,,
\]
For instance, we have $f_3^{(5)}(\vec{x}) = 
x_0 x_1 x_2 +
x_0 x_1 x_4 +
x_0 x_3 x_4 +
x_1 x_2 x_3 +
x_2 x_3 x_4$.

The following lemma was implicitly proved in previous works:
\begin{lemma}[%
  Lyapunov function $V_3$
  {\cite[Page 240, Proof of Theorem~1]{FengZhang14} and \cite[Theorem 4]{Haslegrave14}}%
]\label{lem-drop-f3}
Let $V_3 : Z \to [0, \infty)$ be defined by $V_3(z) := 4 N^2 f_3^{(K)}(\vec{g}(z)/N)$ for $z \in Z_K$.
%Let $z \in Z_K$ for some $k\le K$,
Denote by $z' \in Z_1 \cup Z_3 \cup \ldots \cup Z_K$ 
%$z' \in \bigcup_{i \ge 0} Z_{K-2 i}$
the random successor configuration of~$z \in Z_K$.
Then $\E(V_3(z') \mid z) = V_3(z) - \frac{K-1}{2}$ for all $z \in Z_K$.
Hence, by Proposition~\ref{prop-lyapunov}, $\E \tim_z \le 4 N^2 f_3^{(K)}(\vec{g}(z)/N)$.
\end{lemma}
For $K=3$ Lemma~\ref{lem-drop-f3} gives
$\E \tim_z \le 4 N^2 f_3^{(K)}(\vec{g}(z)/N) = \frac{4}{N} g_0 g_1 g_2$.
In fact, for $K=3$ it was shown before in~\cite{McIverMorgan} that $\E \tim_z$ is identically equal to $\frac{4}{N} g_0 g_1 g_2$,
providing an exact formula for the expected stabilization time of configurations with three tokens.
%(with the worst case $\frac{4}{27} N^2$ being attained for $g_0=g_1=g_2=N/3$).
%
Lemma~\ref{lem-drop-f3} suggests analyzing~$f_3$:
\begin{lemma}[Maximum of $f_3$ {\cite[Proof of Theorem~2]{FengZhang14}, \cite[Theorem 3]{Haslegrave14}}] \label{lem-f3-max}
For all $K \ge 3$ odd we have
\[
 \max_{\vec{x} \in D} f_3^{(K)}(\vec{x}) 
 \ = \ f_3^{(K)}\left(\frac{1}{K}, \ldots, \frac{1}{K}\right)
 \ = \ \frac{1}{24} \left(1 - \frac{1}{K^2} \right)\,.
\]
\end{lemma}
By combining Lemmas \ref{lem-drop-f3}~and~\ref{lem-f3-max} one obtains
$\E \tim_z \le \frac{N^2}{6}(1 - \frac{1}{K^2})$,
which is the bound obtained in~\cite{FengZhang14}.
A slightly better bound is given in~\cite{Haslegrave14}. 

\section{Proof of the Herman Conjecture} \label{sec-overview} 

The function~$V_3$ from Lemma~\ref{lem-drop-f3} leaves room for improvement since $\E(V_3(z') \mid z) = V_3(z) - \frac{K-1}{2}$,
which is strictly less than $V_3(z) - 1$ for $K > 3$.
The idea for obtaining an optimal bound is to decrease the gap between $\frac{K-1}{2}$ and~$1$,
by decreasing the Lyapunov function~$V$.
One could think that the scaled function $\frac{2}{K-1} V_3$ is also a Lyapunov function satisfying~\eqref{eq-prop-lyapunov-step}, but this is not true;
in particular, note that the number of tokens~$K$ might be different for a configuration~$z$ and its successor~$z'$.
Since scaling does not work, we decrease the Lyapunov function by subtracting a quintic polynomial, as follows.
Define a quintic polynomial $f_5^{(K)} : D^{(K)} \to [0, \infty)$,
similar to~$f_3^{(K)}$:
\begin{align*}
  f_5^{(K)}(\vec{x}) =
    \hskip-2em
    \sum_{\substack{0\le i_0<i_1<\cdots<i_4<K\\\text{$i_4-i_3,\ldots,i_1-i_0$ odd}}}
    \hskip-2em
      x_{i_0} x_{i_1} x_{i_2} x_{i_3} x_{i_4}
%f_5^{(K)}(\vec{x}) &:= \sum_{i_0=0}^\infty x_{i_0} \cdot \left( 
%          \sum_{i_1=0}^\infty x_{i_0 + 2 i_1 + 1} \cdot \ldots \cdot \left(
%          \sum_{i_4=0}^\infty x_{i_0 + 2 i_1 + 2 i_2 + 2 i_3 + 2 i_4 + 4} \right) \ldots \right)
\end{align*}
For instance, 
$f_5^{(3)}(\vec{x}) = 0$,
$f_5^{(5)}(\vec{x}) = 
%x_0 x_1 x_2 +
%x_1 x_2 x_3 +
%x_2 x_3 x_4 +
%x_3 x_4 x_0 +
%x_4 x_0 x_1
x_0 x_1 x_2 x_3 x_4$,
and
$f_5^{(7)}(\vec{x}) = 
x_0 x_1 x_2 x_3 x_4 +
x_0 x_1 x_2 x_3 x_6 +
x_0 x_1 x_2 x_5 x_6 +
x_0 x_1 x_4 x_5 x_6 +
x_0 x_3 x_4 x_5 x_6 +
x_1 x_2 x_3 x_4 x_5 +
x_2 x_3 x_4 x_5 x_6$.
%
%x_0 x_1 x_2 x_3 x_4 +
%x_1 x_2 x_3 x_4 x_5 +
%x_2 x_3 x_4 x_5 x_6 +
%x_3 x_4 x_5 x_6 x_0 +
%x_4 x_5 x_6 x_0 x_1 +
%x_5 x_6 x_0 x_1 x_2 +
%x_6 x_0 x_1 x_2 x_3$.
We also define a polynomial $f^{(K)} : D^{(K)} \to [0, \infty)$:
\begin{equation} \label{eq-f} 
f^{(K)}(\vec{x}) := f_3^{(K)}(\vec{x}) - \alpha f_5^{(K)}(\vec{x}) \qquad 
\text{with } \alpha := 24
\end{equation}
For example,
$f^{(5)}(\vec{x}) =
x_0 x_1 x_2 +
x_0 x_1 x_4 +
x_0 x_3 x_4 +
x_1 x_2 x_3 +
x_2 x_3 x_4 - \alpha x_0 x_1 x_2 x_3 x_4$.
%
%x_0 x_1 x_2 +
%x_1 x_2 x_3 +
%x_2 x_3 x_4 +
%x_3 x_4 x_0 +
%x_4 x_0 x_1 - \alpha x_0 x_1 x_2 x_3 x_4$.
Throughout the paper we use $\alpha$ in the expression of $f^{(K)}$ for notational convenience. From now onwards we may drop the superscript $K$  from the domain $D^{(K)}$ of the functions $f_3^{(K)}$, $f_5^{(K)}$ and $f^{(K)}$ to avoid notational clutter when $K$ is understood.

The following properties of~$f$ are fundamental:
\begin{lemma}[Symmetry and continuity properties] \label{lem-f-properties}
The function~$f$ has the following properties.
\begin{itemize}
\item[(a)] It is symmetric with respect to rotation: % and reflection: 
\[ 
f(x_0, \ldots, x_{K-1}) = f(x_1, \ldots, x_{K-1}, x_0) % = f(x_{K-1},\ldots,x_1,x_0).
\]
\item[(b)] It is continuous: %:\footnote{We adopt this term from~\cite{FengZhang14}.}
For $K \ge 5$ we have 
\[
f^{(K)}(x_0, 0, x_2, x_3, \ldots, x_{K-1}) = f^{(K-2)}(x_0+x_2, x_3, \ldots, x_{K-1}).
\]
\end{itemize}
\end{lemma}
Analogous properties were shown for~$f_3$ in~\cite{FengZhang14}.
Their proof carries over to~$f_5$ and hence to~$f$.
The following lemma uses~$f$ to define a tighter Lyapunov function.
\begin{lemma}[Lyapunov function $V$] \label{lem-drop-f}
Define $V : Z \to [0, \infty)$ by
$V(z) := 4 N^2 f(\vec{g}(z)/N)$. 
Let $z \in Z$ and denote by $z'$ the random successor configuration of~$z$.
Then $\E(V(z') \mid z) \le V(z) - 1$.
Hence, by Proposition~\ref{prop-lyapunov}, $\E \tim_z \le 4 N^2 f(\vec{g}(z)/N)$.
\end{lemma}
We remark that
  a similar Lyapunov function has been investigated in~\cite[Equation~(15)]{FengZhang14},
  but did not lead to a proof of the Herman conjecture.
It seems that $V(z)$~needs to be chosen with great care,
  since even slight variations do not work.

Lemma~\ref{lem-drop-f} suggests analyzing~$f$:
\begin{lemma}[Maximum of $f$] \label{lem-f-max}
For all $K \ge 3$ odd we have
\[
 \max_{\vec{x} \in D} f^{(K)}(\vec{x}) = \frac{1}{27}.
 \]
%and the maximum is only attained for vectors~$\vec{x}$ with three entries equal to~$\frac13$
%(and so the other entries are zero).
\end{lemma}
With this in hand our main result follows:
\begin{proof}[Proof of Theorem~\ref{thm-main}]
Immediate by combining Lemmas \ref{lem-drop-f}~and~\ref{lem-f-max}.
\end{proof}
It remains to prove Lemmas~\ref{lem-drop-f}~and~\ref{lem-f-max}.

\subsection{Proof of Lemma~\ref{lem-drop-f}} \label{sub-drop-f} 

Towards Lemma~\ref{lem-drop-f} we show:

\begin{lemma}[Lyapunov function $V_5$] \label{lem-incr-f5}
Define $V_5 : Z \to [0, \infty)$ by $V_5(z) := 4 N^2 f_5(\vec{g}(z)/N)$.
Let $K \ge 5$ and $z \in Z$ and denote by $z'$ the random successor configuration of~$z$.
Then 
\[
\E(V_5(z') \mid z) = V_5(z) + \frac{1}{32} \frac{(K-1)(K-3)}{N^2} - \frac{1}{2} (K - 3) f_3\left(\frac{\vec{g}(z)}{N}\right)\,.
\]
\end{lemma}
The proof in \iftechrep{Appendix~\ref{sec:incr-f5-proof}}{\cite{BGKOW15-techrep}} requires an analysis of correlations among the changes in gaps between tokens in each step of the protocol.
Using Lemma~\ref{lem-incr-f5} one can readily prove Lemma~\ref{lem-drop-f}: 
\begin{proof}[Proof of Lemma~\ref{lem-drop-f}]
For $K=3$ the statement follows from Lemma~\ref{lem-drop-f3}.
For $K \ge 5$ we have: %\maria{isn't this valid for $K=3$ too?}
\begin{align*}
\E(V(z') \mid z) 
& = \E( (V_3(z') - 24 V_5(z') ) \mid z) && \text{by the definitions} \\
& = \E(V_3(z') \mid z) - 24 \E(V_5(z') \mid z) && \text{linearity of expectation} \\
& = V_3(z) - \frac{K-1}{2} - 24 V_5(z) - \frac{3}{4} \frac{(K-1)(K-3)}{N^2} \\
& \ \qquad + 12 (K - 3) f_3\left(\frac{\vec{g}(z)}{N}\right) && \text{Lemmas \ref{lem-drop-f3}~and~\ref{lem-incr-f5}} \\
& \le V(z) - \frac{K-1}{2}  + 12 (K - 3) f_3\left(\frac{\vec{g}(z)}{N}\right) && \text{since $K \ge 3$} \\
& \le V(z) - \frac{K-1}{2} + \frac{K-3}{2} && \text{Lemma~\ref{lem-f3-max}} \\
& = V(z) - 1
\end{align*}
\end{proof}

\subsection{Proof of Lemma~\ref{lem-f-max}}  \label{sub-f-max} 

Towards Lemma~\ref{lem-f-max} we show:

\begin{lemma}[Local maxima of $f$] \label{lem-f-max-induction}
Let $K \ge 5$ and odd.
There is no $\vec{v} \in D^{(K)}$ in the interior of~$D^{(K)}$
 such that $\vec{v}$ is a local maximum and $f^{(K)}(\vec{v}) > \frac{1}{27}$.
\end{lemma}
The proof in Section~\ref{sec-f-max-induction}
  involves a combinatorial analysis of inequalities arising
  from conditions on the derivatives of~$f^{(K)}$.
Using Lemma~\ref{lem-f-max-induction} one can readily prove Lemma~\ref{lem-f-max}:
\begin{proof}[Proof of Lemma~\ref{lem-f-max}]
We proceed by induction on~$K$.
For the induction base we have $K=3$.
It is straightforward to check that the maximum of
$f^{(3)}(\vec{x}) = f_3^{(3)}(\vec{x}) = x_0 x_1 x_2$
is $f^{(3)}(\frac{1}{3},\frac{1}{3},\frac{1}{3}) = \frac{1}{27}$.

For the induction step we have $K \ge 5$.
Let $\vec{v} \in D^{(K)}$ with $f^{(K)}(\vec{v}) = \max_{\vec{x} \in D^{(K)}} f^{(K)}(\vec{x})$.
If $\vec{v}$ is in the interior of~$D^{(K)}$,
then by Lemma~\ref{lem-f-max-induction} we have $f^{(K)}(\vec{v}) \le \frac{1}{27}$.
If $\vec{v}$ is at the boundary of~$D^{(K)}$,
then $v_i = 0$ for some~$i$.
By Lemma~\ref{lem-f-properties}(a) we can assume that $v_1=0$.
Using Lemma~\ref{lem-f-properties}(b) the statement follows from the induction hypothesis.
\end{proof}

\section{Proof of Lemma~\ref{lem-f-max-induction}} \label{sec-f-max-induction}

In this section we prove Lemma~\ref{lem-f-max-induction}.
In Section~\ref{sub-derivatives} we state several properties that an interior local maximum of~$f^{(K)}$ would have to satisfy.
In Section~\ref{sub-proof-5} we prove Lemma~\ref{lem-f-max-induction} for $K=5$ for a first taste of the general argument.
In Section~\ref{sub-proof-7} we prove Lemma~\ref{lem-f-max-induction} for $K=7$
to illustrate some fine points that occur only for larger values of~$K$.
In Section~\ref{sub-combinatorics} we state some combinatorial facts needed for the general case.
Finally, in Section~\ref{sub-proof-f-max-induction} we prove Lemma~\ref{lem-f-max-induction}.

\subsection{Properties of an Interior Local Maximum} \label{sub-derivatives}

The following lemma is obtained by considering first and second derivatives of~$f$
evaluated at an interior local maximum.
\begin{lemma} \label{lem-derivatives}
Let $\vec{v}$ be a local maximum of~$f^{(K)}$ in the interior of~$D^{(K)}$ and define $c \in \R$ by
\begin{align} 
  c \ \ =
    \sum_{\substack{1<i_2<K\\\text{$i_2$ even}}} v_{i_2}
   \  - \quad
   \alpha   \hskip-1.5em
    \sum_{\substack{1<i_2<i_3<i_4<K\\\text{$i_2$, $i_4$ even}\\\text{$i_3$ odd}}}
    \hskip-1.5em
       v_{i_2} v_{i_3} v_{i_4}\,.
  \label{eq-proof-c-def}
  \end{align}
This expression holds for the same value of~$c$ if the indices are rotated by an arbitrary~$k$:
for all~$j$ the index~$i_j$ becomes $(i_j + k) \bmod K$. Further, we have 
\begin{align}
  \sum_{\substack{3\le i_3<i_4<K\\\text{$i_3$ odd}\\\text{$i_4$ even}}} \!\!\!\! v_{i_3} v_{i_4}
  \le
  \frac{1}{\alpha}\,.
  \label{eq-d2-nonpos}
\end{align}
Again, this inequality also holds when indices are rotated.
\end{lemma}
For example, for $K=7$ we have
$c =  v_2 + v_4 + v_6 - \alpha ( v_2 v_3 v_4 + v_2 v_3 v_6 + v_2 v_5 v_6 + v_4 v_5 v_6 )
   =  v_1 + v_3 + v_5 - \alpha ( v_1 v_2 v_3 + v_1 v_2 v_5 + v_1 v_4 v_5 + v_3 v_4 v_5 )$.

\begin{proof}[Proof of Lemma~\ref{lem-derivatives}] The idea of the proof is as follows. We pick a particular direction in $D^{(K)}$,
namely $\vec{d}=(-1,0,1,0,0,\ldots,0)$,
and consider the function $f(\vec{v}+\epsilon \vec{d})$ as a univariate function of~$\epsilon$.
Since $\vec{v}$ is a local maximum, 
the first derivative must be zero and the second derivative must be nonpositive.
Exploiting the fact that $v_i > 0$ for all~$i$ holds in the interior,
we obtain \eqref{eq-proof-c-def}~and~\eqref{eq-d2-nonpos}, respectively. See \iftechrep{Appendix~\ref{app-derivatives}}{\cite{BGKOW15-techrep}} for the detailed proof.
\end{proof}

Let $S_j^{(K)}(\vec{x})$ denote
  the scalar product of $\vec{x}$ with a copy of itself rotated $j$ times:
\[
  S_j^{(K)}(\vec{x}) := \sum_{i=0}^{K-1} x_i x_{i+j}
\]
In all formulas it will be the case that the subscript of~$S$ is odd.
Also, the superscript will be omitted when unimportant or understood from context.

% RG: Maria, S_j is used in the main proof, so it should be defined on its own,
%   not just in the scope of a corollary. (I initially put it in a corollary
%   becaue I was concerned with other things, but then Stefan rightly pulled it out.)
%If $K$ is understood, define for any odd~$j$: \maria{I don't understand this. I would put it inside the Corollary \ref{cor-d2-nonpos-summed}. See below}
%\[
% S_j:=\sum_{i=0}^{K-1} x_i x_{i+j}
%\]

\newcommand{\stmtcordtwononpossummed}{
Let $\vec{v}$ be a local maximum of $f^{(K)}$ in the interior of~$D^{(K)}$.
Then the following inequality holds:
\[
  \sum_{\substack{1 \le i < K-2\\\text{$i$ odd}}} \frac{K-i-2}{2} S_i(\vec{v}) \quad \le \quad \frac{K}{\alpha}
\]
}
\begin{corollary} \label{cor-d2-nonpos-summed}
\stmtcordtwononpossummed
\end{corollary}
For example, for $K=11$ we have
  $4S_1(\vec{v}) + 3S_3(\vec{v}) + 2S_5(\vec{v}) +S_7(\vec{v}) \le 11/\alpha$.

\begin{lemma}[Bound for $f_5$] \label{lem-bound-on-f5}
Suppose that $\vec{v} \in D^{(K)}$ satisfies $f^{(K)}(\vec{v}) > \frac{1}{27}$.
Then $\alpha f_5(\vec{v}) < \frac{1}{216}$.
\end{lemma}
\begin{proof}
By Lemma~\ref{lem-f3-max} we have $f_3(\vec{v}) \le \frac{1}{24}$ and hence
$
 \alpha f_5(\vec{v}) = f_3(\vec{v}) - f(\vec{v}) < \frac{1}{24} - \frac{1}{27} = \frac{1}{216}
$.
\end{proof}

\subsection{Proof of Lemma~\ref{lem-f-max-induction} for $K=5$} \label{sub-proof-5}

Let $K=5$. Then 
\begin{align*}
 f(\vec{x}) \ &= \ f_3(\vec{x}) - \alpha f_5(\vec{x}) 
            \ =  \
x_0 x_1 x_2 +
x_0 x_1 x_4 +
x_0 x_3 x_4 +
x_1 x_2 x_3 +
x_2 x_3 x_4 - \alpha x_0 x_1 x_2 x_3 x_4
%x_0 x_1 x_2 +
%x_1 x_2 x_3 +
%x_2 x_3 x_4 +
%x_3 x_4 x_0 +
%x_4 x_0 x_1 - \alpha x_0 x_1 x_2 x_3 x_4
\end{align*}
Towards a contradiction, suppose that there is a local maximum $\vec{v}$ with $f(\vec{v}) > \frac{1}{27}$ in the interior of~$D$. %, i.e., $v_i > 0$ for all~$i \in \{0, \ldots, 4\}$.
By~\eqref{eq-proof-c-def}, the value
\begin{equation} \label{eq-proof5-c-def}
 c \ = \ v_2 + v_4 - \alpha v_2 v_3 v_4
% c \ = \ v_i + v_{i+2} - \alpha v_i v_{i+1} v_{i+2} \qquad \text{for $i \in \{0, \ldots, 4\}$.}
\end{equation}
is invariant under rotations. Indeed, $ v_{2+k} + v_{4+k} - \alpha v_{2+k} v_{3+k} v_{4+k} \equiv c$ for all~$k$,
  but we shall avoid explicitly mentioning rotations, for notational simplicity.
Summing~\eqref{eq-proof5-c-def} over all $K$~rotations we obtain:
\begin{equation} \label{eq-proof5-five-c}
 5 c \ = \ 2 - \alpha f_3(\vec{v})
\end{equation}
By~\eqref{eq-proof5-c-def} we have 
$v_0 v_1 c =
v_0 v_1 v_2 +
v_0 v_1 v_4 - \alpha f_5(\vec{v})$
%$c v_{i+3} v_{i+4} = 
%v_{i+3} v_{i+4} v_i +
%v_{i+2} v_{i+3} v_{i+4} - \alpha f_5(\vec{v})$
and, summing over all $K$~rotations,
\begin{equation} \label{eq-proof5-c-sum}
 c S_1(\vec{v}) \ = \ 2 f(\vec{v}) - 3 \alpha f_5(\vec{v})
 \end{equation}
Moreover,
\[
 c S_1(\vec{v}) \ \mathop{\le}^\text{Cor.~\ref{cor-d2-nonpos-summed}} \ \frac{5 c}{\alpha} \ \mathop{=}^\text{\eqref{eq-proof5-five-c}} \
 \frac{2}{\alpha} - f_3(\vec{v})  \ = \
 \frac{2}{\alpha} - f(\vec{v}) - \alpha f_5(\vec{v}).
\]
Combining this with~\eqref{eq-proof5-c-sum} gives:
\[
 \frac{2}{\alpha} \quad \ge \quad 3 f(\vec{v}) - 2 \alpha f_5(\vec{v})
 \quad \mathop{\ge}^\text{Lemma~\ref{lem-bound-on-f5}} \quad \frac{3}{27} - 2 \cdot \frac{1}{216}
\]
This implies $\alpha \le 216/11 \approx 19.6$, which is a contradiction as required (since  $\alpha=24$).
\qed

\subsection{Proof of Lemma~\ref{lem-f-max-induction} for $K=7$} \label{sub-proof-7}

Let $K=7$. Towards a contradiction, we suppose again that there is a local maximum $\vec{v}$ with $f(\vec{v}) > \frac{1}{27}$ in the interior of~$D$. %, i.e., $v_i > 0$ for all~$i \in \{0, \ldots, 6\}$.
%By~\eqref{eq-proof-c-def} there is $c \in \R$ such that for all $i \in \{0, \ldots, 6\}$ we have:
By~\eqref{eq-proof-c-def}, all $K$~rotations of the following hold with the same~$c\in\R$:
\begin{equation} \label{eq-proof7-c-def}
 c \ = \ v_2 + v_4 + v_6 - \alpha ( v_2 v_3 v_4 + v_2 v_3 v_6 + v_2 v_5 v_6 + v_4 v_5 v_6 )
\end{equation}
Summing~\eqref{eq-proof7-c-def} over $K$~rotations we obtain:
\begin{equation} \label{eq-proof7-five-c}
 7 c \ = \ 3 - 2 \alpha f_3(\vec{v})
\end{equation}
By~\eqref{eq-proof7-c-def} we have 
\begin{equation} \label{eq-proof7-mult-c1}
 v_0 v_1 c \;=\;
v_0 v_1 v_2 +
v_0 v_1 v_4 +
v_0 v_1 v_6
 - \alpha (
v_0 v_1 v_2 v_3 v_4 +
v_0 v_1 v_2 v_3 v_6 +
v_0 v_1 v_2 v_5 v_6 +
v_0 v_1 v_4 v_5 v_6
)
%
%& c v_{i+5} v_{i+6} = 
%v_{i+5} v_{i+6} v_i +
%v_{i+5} v_{i+6} v_{i+2} + 
%v_{i+4} v_{i+5} v_{i+6} \\
%& \ - \alpha (
%v_{i+5} v_{i+6} v_i v_{i+1} v_{i+2} +
%v_{i+4} v_{i+5} v_{i+6} v_i v_{i+1} + 
%v_{i+3} v_{i+4} v_{i+5} v_{i+6} v_i +
%v_{i+2} v_{i+3} v_{i+4} v_{i+5} v_{i+6}
%)
\end{equation}
and
\begin{equation} \label{eq-proof7-mult-c2}
\begin{aligned}
v_0 v_3 c
 \;&=\;
v_0 v_3 v_4 +
v_0 v_3 v_6
- \alpha v_0 v_3 v_4 v_5 v_6
+ v_0 v_2 v_3 ( 1 - \alpha ( v_3 v_4 + v_3 v_6 + v_5 v_6 ) ) \\
\;&\ge\;
v_0 v_3 v_4 +
v_0 v_3 v_6
- \alpha v_0 v_3 v_4 v_5 v_6
%& \ge v_{i+6} v_i v_{i+3} + v_{i+2} v_{i+3} v_{i+6} - \alpha v_{i+6} v_i v_{i+1} v_{i+2} v_{i+3}\;,
%
%c v_{i+3} v_{i+6} 
%& = 
%v_{i+6} v_i v_{i+3} +
%v_{i+2} v_{i+3} v_{i+6} - \alpha v_{i+6} v_i v_{i+1} v_{i+2} v_{i+3} \\
%& \quad + v_{i+3} v_{i+4} v_{i+6} ( 1 - \alpha ( v_i v_{i+1} + v_i v_{i+3} + v_{i+2} v_{i+3}  ) ) \\
%& \ge v_{i+6} v_i v_{i+3} + v_{i+2} v_{i+3} v_{i+6} - \alpha v_{i+6} v_i v_{i+1} v_{i+2} v_{i+3}\;,
\end{aligned}
\end{equation}
where the last inequality is by~\eqref{eq-d2-nonpos}.
Summing \eqref{eq-proof7-mult-c1}~and~\eqref{eq-proof7-mult-c2} over $K$~rotations we obtain:
\begin{equation} \label{eq-proof7-c-sum}
 c \bigl(2 S_1(\vec{v}) + S_3(\vec{v})\bigr)
  \ \ge \ 
  4 f_3(\vec{v}) - 9 \alpha f_5(\vec{v}) = 4 f(\vec{v}) - 5 \alpha f_5(\vec{v})
\end{equation}
Further we have:
\[
 c \bigl(2 S_1(\vec{v}) + S_3(\vec{v})\bigr)
 \ \mathop{\le}^\text{Cor.~\ref{cor-d2-nonpos-summed}} \ 
 \frac{7 c}{\alpha} \ \mathop{=}^\text{\eqref{eq-proof7-five-c}} \
 \frac{3}{\alpha} - 2 f_3(\vec{v})  \ = \
 \frac{3}{\alpha} - 2 f(\vec{v}) - 2 \alpha f_5(\vec{v})
\]
Combining this with~\eqref{eq-proof7-c-sum} gives:
\[
 \frac{3}{\alpha} \quad \ge \quad 6 f(\vec{v}) - 3 \alpha f_5(\vec{v})
 \quad \mathop{\ge}^\text{Lemma~\ref{lem-bound-on-f5}} \quad \frac{6}{27} - 3 \cdot \frac{1}{216}
\]
This leads to $\alpha \le 14.4$, which is a contradiction as desired.
\qed

\subsection{Combinatorial Lemmas} \label{sub-combinatorics}

In order to generalize the proofs from Sections \ref{sub-proof-5}~and~\ref{sub-proof-7}
 to any odd~$K$, we state some combinatorial lemmas in this subsection.
They are proved in \iftechrep{Appendix~\ref{app-combin-proofs}}{\cite{BGKOW15-techrep}}.

In order to generalize \eqref{eq-proof5-five-c}~and~\eqref{eq-proof7-five-c}
we show the following lemma:

\newcommand{\stmtlemsumcombinsimple}{
We have:
\[
  \sum_{k=0}^{K-1}
  \sum_{\substack{1<i'_0<i'_1<i'_2<K\\\text{$i'_0$, $i'_2$ even}\\\text{$i'_1$ odd}}}
    x_{i'_0+k} x_{i'_1+k} x_{i'_2+k}
  =
  \frac{K-3}{2}
  \hskip-1em
  \sum_{\substack{0\le i_0<i_1<i_2<K\\\text{$i_2-i_1$, $i_1-i_0$ odd}}}
  \hskip-2em
    x_{i_0} x_{i_1} x_{i_2}
  =
  \frac{K-3}{2} f_3^{(K)}(\vec{x})
\]
}
\begin{lemma}\label{lem-sum-combin-simple}
\stmtlemsumcombinsimple
\end{lemma}
For example, if $K=5$, then we obtain that summing the $5$~rotations of
  $x_2 x_3 x_4$
gives $f_3^{(5)}(\vec{x})$.
As another example, if $K=7$, then we obtain that summing the $7$~rotations of
$
  x_2 x_3 x_4
  + x_2 x_3 x_6
  + x_2 x_5 x_6
  + x_4 x_5 x_6
$
gives $2 f_3^{(7)}(\vec{x})$.
These two instances of Lemma~\ref{lem-sum-combin-simple}
  help establish \eqref{eq-proof5-five-c}~and~\eqref{eq-proof7-five-c}.

In order to generalize the inequality in~\eqref{eq-proof7-mult-c2}
we need the following lemma:
\newcommand{\stmtlemdropterms}{
Let $\vec{v}$ be a local maximum of $f^{(K)}$ in the interior of~$D^{(K)}$.
If $i_1$ is odd and $0<i_1<K$, then the following inequality holds:
\[
v_0 v_{i_1}
  \biggl(
    \sum_{\substack{1<i_2<K\\\text{$i_2$ even}}}
      v_{i_2}
    -
    \hskip-1em
    \sum_{\substack{1<i_2<i_3<i_4<K\\\text{$i_2,i_4$ even}\\\text{$i_3$ odd}}}
    \hskip-2em
      \alpha v_{i_2} v_{i_3} v_{i_4}
  \biggr)
  \ge
v_0 v_{i_1}
  \biggl(
    \sum_{\substack{i_1<i_2<K\\\text{$i_2$ even}}}
      v_{i_2}
    -
    \hskip-1em
    \sum_{\substack{i_1<i_2<i_3<i_4<K\\\text{$i_2,i_4$ even}\\\text{$i_3$ odd}}}
    \hskip-2em
      \alpha v_{i_2} v_{i_3} v_{i_4}
  \biggr)
\]
}
\begin{lemma}\label{lem-drop-terms}
\stmtlemdropterms
\end{lemma}
The inequality says that if we drop those terms that do not occur in $f_3^{(K)}$ or $f_5^{(K)}$,
  then we obtain a lower bound.
The proof groups those terms that are not in either of $f_3^{(K)}$ or $f_5^{(K)}$,
  and then invokes~\eqref{eq-d2-nonpos} to show that their sum is nonnegative.

In order to generalize \eqref{eq-proof5-c-sum}~and~\eqref{eq-proof7-c-sum}
we need Corollary~\ref{cor-fancy-sum} below,
which is a consequence of the following lemma:
\newcommand{\stmtlemfancysum}{
Let $l$ be an odd, positive integer.
Then:
\begin{multline*}
  \sum_{k=0}^{K-1}
  \sum_{\substack{1\le i'_1< K-2\\\text{$i'_1$ odd}}}
    \frac{K-i'_1-2}{2}
  \hskip-1em
  \sum_{\substack{%
    i'_1<i'_2<\cdots<i'_{l-1}<K\\
    \forall j,\; i'_j \equiv j \pmod 2
  }}
    x_k x_{i'_1+k}
    \prod_{\substack{1<j<l}}
      x_{i'_j+k}
  =
\\
  =
  \Bigl( \frac{l-1}{2} K - l \Bigr)
  \hskip-1em
  \sum_{\substack{
    0\le i_0<\cdots<i_{l-1}<K\\
    \text{$i_j-i_{j-1}$ odd for $0<j<l$}}}
  \prod_{j=0}^{l-1}
    x_{i_j}
\end{multline*}
}
\begin{lemma}\label{lem-fancy-sum}
\stmtlemfancysum
\end{lemma}
For example, if $K=5$ and $l=3$, then we have that summing $5$~rotations of
$
  x_0 x_1 x_2 + x_0 x_1 x_4
$
gives $2f_3^{(5)}(\vec{x})$.
As another example, if $K=9$ and $l=3$, then summing $9$~rotations of
$
  3x_0x_1(x_2+x_4+x_6+x_8) + 2x_0x_3(x_4+x_6+x_8) + x_0x_5(x_6+x_8)
$
gives $6 f_3^{(9)}(\vec{x})$.

\begin{corollary}\label{cor-fancy-sum}
We have:
\[
  \sum_{k=0}^{K-1}
  \sum_{\substack{1 \le i_1 < K-2\\\text{$i_1$ odd}}}
     \frac{K-i_1-2}{2}
  \sum_{\substack{i_1<i_2<K\\\text{$i_2$ even}}}
    x_{0+k} x_{i_1+k} x_{i_2+k}
  =
  (K-3) f_3^{(K)}(\vec{x})
\]
and also
\[
  \sum_{k=0}^{K-1}
  \sum_{\substack{1 \le i_1 < K-2\\\text{$i_1$ odd}}}
     \frac{K-i_1-2}{2}
  \sum_{\substack{i_1<i_2<i_3<i_4<K\\\text{$i_2,i_4$ even}\\\text{$i_3$ odd}}}
  x_{0+k} x_{i_1+k}
  x_{i_2+k} x_{i_3+k} x_{i_4+k}
  =
  (2K-5) f_5^{(K)}(\vec{x})
\]
\end{corollary}

\begin{proof}
Instantiate Lemma~\ref{lem-fancy-sum} with $l=3$ and, respectively, $l=5$.
\end{proof}

\subsection{Proof of Lemma~\ref{lem-f-max-induction}} \label{sub-proof-f-max-induction}

Towards a contradiction, suppose that there is a local maximum $\vec{v}$ with $f(\vec{v}) > \frac{1}{27}$ in the interior of~$D$, i.e., $v_i > 0$ for all~$i \in \{0, \ldots, K-1\}$.
Summing up the $K$ rotations of~\eqref{eq-proof-c-def}  and using Lemma~\ref{lem-sum-combin-simple}, we obtain:
\begin{equation} \label{eq-proof-Mc}
 K c \ = \ \frac{K-1}{2} - \frac{K-3}{2} \alpha f_3(\vec{v})
\end{equation}
Multiplying \eqref{eq-proof-c-def} on both sides by
$
    \sum_{\substack{1 \le i_1 < K-2\\\text{$i_1$ odd}}} \frac{K-i_1-2}{2} v_0 v_{i_1}
$
we obtain:
\begin{align*}
c  \sum_{\substack{1 \le i_1 < K-2\\\text{$i_1$ odd}}}  \frac{K-i_1-2}{2} v_0 v_{i_1}  &=
  \sum_{\substack{1 \le i_1 < K-2\\\text{$i_1$ odd}}} \frac{K-i_1-2}{2} v_0 v_{i_1}
  \Biggl(
    \sum_{\substack{1<i_2<K\\\text{$i_2$ even}}} v_{i_2}
    -
    \hskip-1em
    \sum_{\substack{1<i_2<i_3<i_4<K\\\text{$i_2$, $i_4$ even}\\\text{$i_3$ odd}}}
    \hskip-1.5em
      \alpha v_{i_2} v_{i_3} v_{i_4}
  \Biggr) \\
  & \ge
  \sum_{\substack{1 \le i_1 < K-2\\\text{$i_1$ odd}}} \frac{K-i_1-2}{2} v_0 v_{i_1}
  \Biggl(
    \sum_{\substack{i_1<i_2<K\\\text{$i_2$ even}}} v_{i_2}
    -
    \hskip-1em
    \sum_{\substack{i_1<i_2<i_3<i_4<K\\\text{$i_2$, $i_4$ even}\\\text{$i_3$ odd}}}
    \hskip-1.5em
      \alpha v_{i_2} v_{i_3} v_{i_4}
  \Biggr)
\end{align*}
using Lemma~\ref{lem-drop-terms}. Summing $K$~rotations of this inequality yields:
\begin{equation} \label{eq-proof-c-sum}
\begin{aligned}
  c \!\!  \sum_{\substack{1 \le i_1 < K-2\\\text{$i_1$ odd}}} \frac{K-i_1-2}{2} S_{i_1}(\vec{v})
  \ &\ge\
  (K-3) f_3(\vec{v}) - (2K-5) \alpha f_5(\vec{v}) = (K-3) f(\vec{v}) - (K-2) \alpha f_5(\vec{v}).
\end{aligned}
\end{equation}
using Corollary~\ref{cor-fancy-sum}. Further we have:
\[
  c  \sum_{\substack{1 \le i_1 < K-2\\\text{$i_1$ odd}}} \frac{K-i_1-2}{2} S_{i_1}(\vec{v})
  \ \mathop{\le}^\text{Cor.~\ref{cor-d2-nonpos-summed}}\
  \frac{K c}{\alpha}
  \ \mathop{=}^\text{\eqref{eq-proof-Mc}}\
  \frac{K-1}{2\alpha}-\frac{K-3}{2} f_3(\vec{v})
\]
Combining this with~\eqref{eq-proof-c-sum} gives:
\[
  \frac{K-1}{2\alpha}
    \ \ge\
  \frac{3K-9}{2} f(\vec{v}) - \frac{K-1}{2} \alpha f_5(\vec{v})
    \ \mathop{\ge}^\text{Lemma~\ref{lem-bound-on-f5}}\
  \frac{K-3}{2}\cdot \frac{1}{9} - \frac{K-1}{2}\cdot \frac{1}{216}
\]
This implies
\[
  \alpha \le \frac{216(K-1)}{23K-71} < 19.7
\]
Since $\alpha=24$, this leads to a contradiction as desired.
\qed

\section{Conclusions} \label{sec-conclusion}

In this paper we have proved the Herman-Protocol Conjecture formulated by
McIver and Morgan in~\cite{McIverMorgan} a decade ago, which says that the
worst-case initial configuration consists of three maximally-separated tokens,
for $N$~multiple of~$3$. This follows from our result that the worst-case
self-stabilization time is at most $\frac{4}{27}N^2$, for any number of
processes~$N$ and any odd number of tokens~$K$.

% RG: rephrased above
%In this paper we have proved the Herman-Protocol Conjecture posed by
%McIver and Morgan in \cite{McIverMorgan} a decade ago, namely that the
%worst-case expected self-stabilization time of a ring of $N$ processes
%is $\frac{4}{27} N^2$ regardless of the number of tokens
%$K$. Interestingly, this value is equal to the expected
%self-stabilization time of a three-token configuration only, where the
%three tokens are initially equally spaced on the ring.

The proof uses a Lyapunov function approach. To do so, we first find a
suitable Lyapunov function and then show that its maximum
is~$\frac{4}{27} N^2$. Then we show that this function gives an upper
bound for the self-stabilization time for \emph{each} possible
configuration in Herman's algorithm. 

% \maria{I WOULDN'T PUT THIS: 
% Lemma~\ref{lem-incr-f5} shows that our candidate Lyapunov function
%   has the required properties.
% We then analyze the first and second derivatives of the Lyapunov function
%   along certain directions.
% The conditions that that the first derivatives are zero
%   and the second derivatives are nonpositive
%   imply upper and lower bounds on a quantity,
%   $c \sum_{\text{$i$ odd}}a_i S_i$.
% This quantity can be seen as a linear combination of autocorrelations~$S_i$,
%   using weights~$a_i$.
% By transitivity, the lower bound should be lower than the upper bound.
% But, this cannot happen if the Lyapunov function evaluates to more than $4/27 N^2$.
% Since the value $4/27 N^2$ is reached on the boundary of the domain of definition,
%   it follows that the Lyapunov function is~$\le\frac{4}{27}N^2$.}

%A natural variant of the Herman algorithm is its asynchronous version:
%at each time step, one token is chosen and passed on to the clockwise neighbor. We conjecture that in this case the expected self-stabilization time is bounded  by $\frac{1}{9}N^2$ as long as $N\ge 8$.

\subparagraph*{Acknowledgements}

Stefan Kiefer is supported by a University Research Fellowship of the Royal Society and by EPSRC grant EP/M003795/1.
Jo\"el Ouaknine is supported by ERC grant AVS-ISS (648701).

%---------------------
\bibliographystyle{plain} %oder alpha oder splncs
\bibliography{db}
%---------------------

\iftechrep{
\newpage
\appendix

\appendix
\section{Proof of Lemma~\ref{lem-incr-f5}}
\label{sec:incr-f5-proof}

Let $z:\{1,\ldots,K\}\rightarrow\{1,\ldots,N\}$ be a $K$-token
configuration on a ring with $N$ processes.  Recall that the
associated gap vector $\g(z) = (g_0,\ldots,g_{K-1}) \in \mathbb{N}^K$ is
defined by $g_0:=N+z(1)-z(K)$ and $g_i:=z(i+1)-z(i)$ for
$i=1,\ldots,K-1$.

Given $z$, consider the \emph{gap-increment vector}
$\Delta:=\g(z')-\g(z)$, where $z'$ is the random successor
configuration of $z$.  This is a random variable taking values in
$\{-1,0,+1\}^K$ where, for each $i\in\{0,\ldots,K-1\}$, $\Delta_i=0$
with probability $1/2$ (the two tokens adjacent to the $i$-th gap both
stay or both move clockwise), and $\Delta_i = \pm 1$ with probability
$1/4$ (one token stays and the other moves clockwise).

We will need the following two properties (\ref{eq:property1}) and
(\ref{eq:property2}) concerning the expectation of the random
variable $\Delta$.  First, it is straightforward
to verify by direct calculation that for $0\leq k< K$,
\begin{gather}
  \E (\Delta_i\Delta_{i+1}\ldots \Delta_{i+k}) =
\left\{ \begin{array}{ll}
        0 & \mbox{if $k$ is even}\\
        \left(-\frac{1}{4}\right)^{(k+1)/2} & \mbox{if $k$ is odd}
\end{array} \right .
\label{eq:property1} 
\end{gather}
Secondly, suppose that $0\leq i_1 \leq i_2 < i_3 \leq i_4 < K$, with
$i_3 \not\equiv i_2+1$ and $i_1 \not\equiv i_4+1$ modulo $K$, that is,
$\{i_1,\ldots,i_2\}$ and $\{i_3,\ldots,i_4\}$ form two non-adjacent
intervals (treating $K-1$ and $0$ as adjacent).  Then
\begin{gather}
 \E(\Delta_{i_1}\ldots \Delta_{i_2}
     \Delta_{i_3}\ldots \Delta_{i_4})
  =  \E(\Delta_{i_1}\ldots \Delta_{i_2})
     \E(\Delta_{i_3}\ldots \Delta_{i_4}) \, .
\label{eq:property2}
\end{gather}
because $\Delta_{i_1},\ldots,\Delta_{i_2}$ and
$\Delta_{i_3},\ldots,\Delta_{i_4}$ are determined by the movements of
disjoints sets of tokens, and hence are independent.

For a given configuration $z$ we want to compute $\E
[f_5(\g(z)+\Delta)]$.  From the definition of $f_5$ and the linearity
of expectation, this is a sum of expressions of the form
\begin{align}
  \E
    (g_{i_0}+\Delta_{i_0})
    (g_{i_1}+\Delta_{i_1})
    (g_{i_2}+\Delta_{i_2})
    (g_{i_3}+\Delta_{i_3})
    (g_{i_4}+\Delta_{i_4})
\label{eq:expand}
\end{align}
over the set of indices $0\leq i_0<i_1<i_2<i_3<i_4<K$ of alternating parity.

Expression (\ref{eq:expand}) evaluates to a degree-5 polynomial in the
variables $\g$.  Observe that this polynomial has no monomials of even degree.
For example,
  all degree-2 monomials have coefficients of the form $\E(\Delta_i\Delta_j\Delta_k)$
  with $i<j<k$.
These coefficients are zero by \eqref{eq:property1}~and~\eqref{eq:property2}.
Degrees 0 and 4 are proved similarly.
%  Indeed,
%expanding (\ref{eq:expand}) yields degree-4 monomials with
%coefficient $\E(\Delta_i)$ and degree-2 monomials with
%coefficient $\E(\Delta_i\Delta_j\Delta_k)$ with $i<j<k$.  But
%$\E(\Delta_i)$ and $\E(\Delta_i\Delta_j\Delta_k)$
%are both zero by (\ref{eq:property1}) and
%(\ref{eq:property2}).

There is a single degree-5 monomial in (\ref{eq:expand})---namely
$g_{i_0}\ldots g_{i_4}$.  Summing all such terms over indices $0\leq
i_0<i_1<i_2<i_3<i_4<K$ of alternating parity yields $f_5(\g(z))$.

Expanding the expression (\ref{eq:expand}) yields degree-3 monomials
of the form \[g_{j_0}g_{j_1}g_{j_2}\E(\Delta_{j_3}\Delta_{j_4})\] for
distinct indices $j_0<j_1<j_2$.  The coefficient of such a term is
$-1/4$ if $j_4 \equiv j_3+1$ or $j_3\equiv j_4+1$ and $0$ otherwise.
Moreover, if $j_0,j_1,j_2$ have alternating parity there are $(K-3)/2$
choices of $j_3$ such that
$g_{j_0}g_{j_1}g_{j_2}\E(\Delta_{j_3}\Delta_{j_3+1})$ appears in
(\ref{eq:expand}).
If $j_0,j_1,j_2$ do not have alternating parity
then there are no such terms in (\ref{eq:expand}).  We conclude that
the sum of all degree-3 monomials in $\E(f_5(\g(z)+\Delta)$ is
\[ -\frac{(K-3)}{8} f_3(\g(z)) \, . \]

Finally, consider the degree-1 monomials.  These have the form
\[g_{j_0}\E(\Delta_{j_1}\Delta_{j_2}\Delta_{j_3}\Delta_{j_4})\] for
distinct indices $j_0$ and $j_1<j_2<j_3<j_4$.  By Property
(\ref{eq:property2}), such terms are only non-zero if
$\{j_1,j_2,j_3,j_4\}$ comprises either a single block of adjacent
indices or two non-adjacent blocks of length $2$ (considering $K-1$
and $0$ to be adjacent).  In this case
$\E(\Delta_{j_1}\Delta_{j_2}\Delta_{j_3}\Delta_{j_4})= 1/16$, and
there are $\binom{(K-1)/2}{2}=(K-1)(K-3)/8$ such choices of
$\{j_1,j_2,j_3,j_4\}$ for each choice of $j_0$.  Thus $g_{j_0}$ has
total coefficient $(K-1)(K-3)/128$ in $\E(f_5(\g(z)+\Delta))$.
Moreover, since $g_0+\ldots+g_{K-1}=N$, the degree-1 terms in
$\E(f_5(\g(z)+\Delta))$ sum to $N(K-1)(K-3)/128$.

In summary, we have proved:
\begin{proposition}
For each $K$-token configuration $z$,
\[ \E(f_5(\g(z)+\Delta)) = f_5(\g(z)) - \frac{K-3}{8}f_3(\g(z)) +
\frac{(K-1)(K-3)N}{128} \, .\]
\label{prop:lemma8}
\end{proposition}

Lemma~\ref{lem-incr-f5} follows immediately from
Proposition~\ref{prop:lemma8} by scaling, since $V_5(z) = 4N^2
f_5(\g(z)/N)=\frac{4}{N^3}f_5(\g(z))$ is a linear multiple of
$f_5(\g(z))$.

\section{Proofs of Properties of an Interior Local Maximum} \label{app-derivatives}

\begin{qlemma}{\ref{lem-derivatives}}
Let $\vec{v}$ be a local maximum of~$f^{(K)}$ in the interior of~$D^{(K)}$ and define $c \in \R$ by
\begin{align} 
  c \ \ =
    \sum_{\substack{1<i_2<K\\\text{$i_2$ even}}} v_{i_2}
   \  - \quad
   \alpha   \hskip-1.5em
    \sum_{\substack{1<i_2<i_3<i_4<K\\\text{$i_2$, $i_4$ even}\\\text{$i_3$ odd}}}
    \hskip-1.5em
       v_{i_2} v_{i_3} v_{i_4}\,.
  \label{eq-proof-c-def_p}
  \end{align}
This expression holds for the same value of~$c$ if the indices are rotated by an arbitrary~$k$:
for all~$j$ the index~$i_j$ becomes $(i_j + k) \bmod K$. Further, we have 
\begin{align}
  \sum_{\substack{3\le i_3<i_4<K\\\text{$i_3$ odd}\\\text{$i_4$ even}}} \!\!\!\! v_{i_3} v_{i_4}
  \le
  \frac{1}{\alpha}
  \label{eq-d2-nonpos_p}
\end{align}
Again, this inequality also holds when indices are rotated.
\end{qlemma}

\begin{proof}
We consider the second-order Taylor expansion of~$f^{(K)}$ along the direction $\vec{d}=(-1,0,1,0,\ldots,0)$
 (which is tangent to~$D^{(K)}$):
\[
  f^{(K)}(\vec{x}+\epsilon\vec{d}) =
    f^{(K)}(\vec{x}) + \epsilon Q(\vec{x}) + \epsilon^2 R(\vec{x}) + O(\epsilon^3)
  \;.
\]
Since $\vec{v}$ is a local maximum, we have $Q(\vec{v})=0$ and $R(\vec{v})\le 0$.
Proving~\eqref{eq-proof-c-def_p} boils down to calculating~$Q(\vec{x})$;
  proving~\eqref{eq-d2-nonpos_p} boils down to calculating~$R(\vec{x})$.
%In what follows,
%  it helps to think of $x_0,\ldots,x_{K-1}$ as symbolic variables used in formal manipulations.
%On the other hand, $\vec v$~is one particular point in $D^{(K)}$,
%  the assumed local maximum from the lemma statement.

\medskip
First, we prove~\eqref{eq-proof-c-def_p}.
Let
\[
  f^{(K)}(x_0+\epsilon,x_1,\ldots,x_{K-1}) =
    f^{(K)}(\vec{x}) + \epsilon P(\vec{x}) + O(\epsilon^2).
\]
By the chain rule and using the rotational symmetry of~$f^{(K)}$ (Lemma~\ref{lem-f-properties}(a)), we find that
\[
  Q(x_0,\ldots,x_{K-1})=P(x_{0+2},x_{1+2},\ldots,x_{K-1+2})-P(x_0,x_1,\ldots,x_{K-1}).
\]
Now recall the definition of~$f^{(K)}$:
\begin{equation}
\label{fmnew}
  f^{(K)}(\vec{x}) :=
    \hskip-1em
    \sum_{\substack{0\le i_0<i_1<i_2<K\\\text{$i_2-i_1$ and $i_1-i_0$ odd}}}
    \hskip-2em
      x_{i_0} x_{i_1} x_{i_2}
    -
    \hskip-2em
    \sum_{\substack{0\le i_0<\cdots<i_4<K\\\text{$i_4-i_3,\ldots,i_1-i_0$ all odd}}}
    \hskip-2em
      \alpha
      x_{i_0} x_{i_1} x_{i_2} x_{i_3} x_{i_4}
\end{equation}
We differentiate $f^{(K)}$ with respect to $x_0$ to obtain
\begin{align}
  P(\vec{x}) =
    \sum_{\substack{0<i_1<i_2<K\\\text{$i_1$ odd}\\\text{$i_2$ even}}}
    \hskip-1em
      x_{i_1} x_{i_2}
    -
    \hskip-1em
    \sum_{\substack{0<i_1<\cdots<i_4<K\\\text{$i_1,i_3$ odd}\\\text{$i_2,i_4$ even}}}
    \hskip-2em
      \alpha
      x_{i_1} x_{i_2} x_{i_3} x_{i_4}.
\end{align}
Since $\vec v$~is a local maximum, $Q(\vec{v})=0$ and we have
\[
    \sum_{\substack{0<i_1<i_2<K\\\text{$i_1$ odd}\\\text{$i_2$ even}}}
    \hskip-1em
      v_{i_1} v_{i_2}
    -
    \hskip-1em
    \sum_{\substack{0<i_1<\cdots<i_4<K\\\text{$i_1,i_3$ odd}\\\text{$i_2,i_4$ even}}}
    \hskip-2em
      \alpha
      v_{i_1} v_{i_2} v_{i_3} v_{i_4}
=
    \sum_{\substack{0<i_1<i_2<K\\\text{$i_1$ odd}\\\text{$i_2$ even}}}
    \hskip-1em
      v_{i_1+2} v_{i_2+2}
    -
    \hskip-1em
    \sum_{\substack{0<i_1<\cdots<i_4<K\\\text{$i_1,i_3$ odd}\\\text{$i_2,i_4$ even}}}
    \hskip-2em
      \alpha
      v_{i_1+2} v_{i_2+2} v_{i_3+2} v_{i_4+2}
\]
Observe that the monomials not containing $v_1$ cancel each other out.
Dividing by $v_1=v_{1+K}$ (since $v_1>0$), we have
\[
    \sum_{\substack{1<i_2<K\\\text{$i_2$ even}}}
    \hskip-1em
      v_{i_2}
    -
    \hskip-1em
    \sum_{\substack{1<i_2<i_3<i_4<K\\\text{$i_3$ odd}\\\text{$i_2,i_4$ even}}}
    \hskip-2em
      \alpha
      v_{i_2} v_{i_3} v_{i_4}
=
    \sum_{\substack{0<i_1<K-1\\\text{$i_1$ odd}}}
    \hskip-1em
      v_{i_1+2}
    -
    \hskip-1em
    \sum_{\substack{0<i_1<i_2<i_3<K-1\\\text{$i_1,i_3$ odd}\\\text{$i_2$ even}}}
    \hskip-2em
      \alpha
      v_{i_1+2} v_{i_2+2} v_{i_3+2}
\]
Now we observe that the right hand side can be obtained from the left hand side
  by changing each index~$i$ into~$i+1$.
Taking into account rotations of the above equality, we conclude~\eqref{eq-proof-c-def_p}.

\medskip
Next, we prove~\eqref{eq-d2-nonpos_p}. To do so, we first calculate the terms of order $\epsilon^2$
  of $f^{(K)}(\vec{x}+\epsilon\vec{d})$.
Such terms occur only when $i_0=0$,\ $i_1=1$, and $i_2=2$.
 % DON'T EVER TYPE SOMETHING LIKE $i_0=0,i_1=1$ (quote, TeXBook, page 161)
In this case,
  the first sum reduces to $(x_0-\epsilon)x_1(x_2+\epsilon)$,
  and the second sum reduces to
\[
  \alpha
  (x_0-\epsilon)x_1(x_2+\epsilon)
  \hskip-1em
  \sum_{\substack{2<i_3<i_4<K\\\text{$i_3$ odd, $i_4$ even}}}
  \hskip-1em
    x_{i_3} x_{i_4}.
\]
Thus,
\[
  R(\vec{x}) =
  -x_1
  +\alpha x_1
  \hskip-1em
  \sum_{\substack{2<i_3<i_4<K\\\text{$i_3$ odd, $i_4$ even}}}
  \hskip-1em
    x_{i_3} x_{i_4}.
\]
Since the assumed interior local maximum $\vec{v}$ is in the interior of~$D^{(K)}$,
  we have that $v_1>0$,
  and so the condition $R(\vec{v})\le0$ is equivalent to $R({\vec{v}})/v_1\le0$:
\[
  -1 +
  \hskip-1em
  \sum_{\substack{2<i_3<i_4<K\\\text{$i_3$ odd, $i_4$ even}}}
  \hskip-1.5em
    \alpha v_{i_3} v_{i_4}
  \le0
\]
Up to trivial rearrangement, we obtained~\eqref{eq-d2-nonpos_p}.
\end{proof}

\begin{qcorollary}{\ref{cor-d2-nonpos-summed}}
\stmtcordtwononpossummed
\end{qcorollary}

\begin{proof}
We sum \eqref{eq-d2-nonpos_p} over all $K$ rotations:
\begin{align*}
  \frac{K}{\alpha}
    &\ge \sum_{i=0}^{K-1} \sum_{\substack{3\le i_3<i_4< K\\\text{$i_3$ odd}\\\text{$i_4$ even}}}
      \!\!\!v_{i_3+i} v_{i_4+i}
     = \sum_{\substack{3\le i_3<i_4< K\\\text{$i_3$ odd}\\\text{$i_4$ even}}} \sum_{i=0}^{K-1}
      v_{i_3+i} v_{i_4+i}
     = \sum_{\substack{3\le i_3<i_4< K\\\text{$i_3$ odd}\\\text{$i_4$ even}}}
        \!\!\! S_{i_4-i_3}(\vec{v}) \\
    &= \sum_{\substack{3\le i_3<i_3+i<K\\\text{$i$ odd, $i_3$ odd}}} \!\!\!\!S_i(\vec{v})
     = \sum_{\substack{1\le i\\\text{$i$ odd}}} \sum_{\substack{3\le i_3<K-i\\\text{$i_3$ odd}}}
       \!\!\!\!S_i(\vec{v})
     = \sum_{\substack{1\le i<K-2\\\text{$i$ odd}}} \frac{K-i-2}{2} S_i(\vec{v})
\end{align*}
\end{proof}

\section{Proofs of Combinatorial Lemmas}\label{app-combin-proofs}

We repeat the combinatorial facts of Section~\ref{sub-combinatorics},
  this time with proofs.

\begin{qlemma}{\ref{lem-sum-combin-simple}}
\stmtlemsumcombinsimple
\end{qlemma}

\begin{proof}
Let us fix $0\le i_0<i_1<i_2<K$ such that both $i_2-i_1$ and $i_1-i_0$ are odd.
We want to show that the term $x_{i_0}x_{i_1}x_{i_2}$
  occurs $(K-3)/2$ times in each side of the equality.
The middle and right sides are trivial; it remains to check the left side.
Let us now fix an arbitrary~$k$.
For a term on the left hand side to equal $x_{i_0}x_{i_1}x_{i_2}$,
  it must be that the sets $\{(i'_0+k)\bmod K, (i'_1+k)\bmod K, (i'_2+k)\bmod K\}$
  and $\{i_0,i_1,i_2\}$ are equal.
In other words, once $i_0,i_1,i_2,k$ are fixed,
  the set $\{i'_0,i'_1,i'_2\}$ is uniquely determined.
Since, $i'_0<i'_1<i'_2$, the potential values of $i'_0,i'_1,i'_2$ are also uniquely determined.
The remaining question is for how many $k\in\{0,\ldots,K-1\}$
  it is the case that the values $i'_0,i'_1,i'_2$ so determined obey
  the other constraints.

There are three disjoint cases:
The smallest value in the set $\{i'_0,i'_1,i'_2\}$, namely~$i'_0$, is
  $(i_0-k)\bmod K$ or $(i_1-k)\bmod K$ or $(i_2-k)\bmod K$.
The case $i'_0=(i_1-k)\bmod K$ occurs exactly when
  (a)~$i_0<k<i_1$, and (b)~$k$~has the same parity as~$i_1$.
Let $\delta_0$ denote the size of the gap between $i_0$~and~$i_1$.
Then, there are $(\delta_0-1)/2$ values of $k$ that obey both (a)~and~(b).
The other two cases are similar,
  and so we conclude that the term $x_{i_0}x_{i_1}x_{i_2}$
  occurs on the left hand side
\[
  \frac{\delta_0-1}{2}+\frac{\delta_1-1}{2}+\frac{\delta_2-1}{2} = \frac{K-3}{2}
\]
times.
\end{proof}

\begin{qlemma}{\ref{lem-drop-terms}}
\stmtlemdropterms
\end{qlemma}

\begin{proof}
The proof below is a case analysis of where $i_1$ can be inserted in-between
  $0<i_2<i_3<i_4<K$.
As noted before, the task is to show that retaining the terms that occur in $f_5$
  gives a lower bound.
In other words, we want to show that those terms not occurring in $f_5$ have a positive sum.
We calculate this sum:
\begin{equation*}
\def\1{\hskip-1em}
\def\2{\hskip-2em}
\begin{aligned}
%  &\sum_{1<i_2<i_1<K} v_0 v_{i_2} v_{i_1}
%  -\1\sum_{1<i_2<i_1\le i_3<i_4<K}\2 \alpha v_0 v_{i_2} v_{i_1} v_{i_3} v_{i_4}
%  -\1\sum_{1<i_2<i_3<i_4<i_1<K}\2 \alpha v_0 v_{i_2} v_{i_3} v_{i_4} v_{i_1} \\
%=\;
  &\sum_{\substack{1<i_2<i_1\\\text{$i_2$ even}}} v_0 v_{i_2} v_{i_1}
  -\1\sum_{\substack{1<i_2<i_3<i_4<K\\i_2<i_1}}\2 \alpha v_0 v_{i_2} v_{i_1} v_{i_3} v_{i_4} \\
=\;
  &\sum_{\substack{1<i_2<i_1\\\text{$i_2$ even}}} v_0 v_{i_2} v_{i_1}
  \Bigl(
  1-
  \1\sum_{\substack{i_2<i_3<i_4<K\\\text{$i_3$ odd, $i_4$ even}}}\1 \alpha v_{i_3} v_{i_4}
  \Bigr) \\
\ge\;
  &\sum_{\substack{1<i_2<i_1\\\text{$i_2$ even}}} v_0 v_{i_2} v_{i_1}
  \Bigl(
  1-
  \1\sum_{\substack{3\le i_3<i_4<K\\\text{$i_3$ odd, $i_4$ even}}}\1 \alpha v_{i_3} v_{i_4}
  \Bigr) \quad\mathop{\ge}^\text{by~\eqref{eq-d2-nonpos}}\quad 0
\end{aligned}
\end{equation*}
%(Above, the constraints `$i_2,i_4$ even' and `$i_3$ odd' are omitted to avoid clutter.)
\end{proof}

\begin{qlemma}{\ref{lem-fancy-sum}}
\stmtlemfancysum
\end{qlemma}

The proof is similar to that of Lemma~\ref{lem-sum-combin-simple}.

\begin{proof}
Let us fix $0\le i_0 <\cdots<i_{l-1}<K$ with odd gaps in-between.
We want to show that the term $x_{i_0}\ldots x_{i_{l-1}}$
  occurs $(l-1)K/2-l$ times on each side of the equation.
For the right side, it is trivial.
The general form of a term on the left side is
$
  x_k x_{i'_1+k} x_{i'_2+k} \ldots x_{i'_{l-1}+k}
$.
It must be that $k$ is one of $i_0,\ldots,i_{l-1}$.
Let us consider the case $k=i_0$; the others are similar.
If $k=i_0$, then, in fact,
\[
  (k, i'_1+k, i'_2+k,\ldots,i'_{l-1}+k) \bmod K \quad=\quad (i_0,i_1,\ldots,i_{l-1}).
\]
In particular, $i'_1$ equals the size of the gap between $i_0$~and~$i_1$.
Let us denote this gap by~$\delta_0$.
On the left hand side, the term is multiplied by $(K-i'_1-2)/2$,
  which is $(K-\delta_0-2)/2$.
The cases $k=i_1$, $k=i_2$, \dots are similar.
Because $\delta_0+\cdots+\delta_{l-1}=K$,
  we conclude that the term $x_{i_0}\ldots x_{i_{l-1}}$ occurs
\[
  \frac{K-\delta_0-2}{2} + \cdots + \frac{K-\delta_{l-1}-2}{2}
  = \frac{l-1}{2}K-l
\]
times on the left side.
\end{proof}

}{}

\end{document}